\documentclass[format=acmsmall, review=false, screen=true]{acmart}
\usepackage{microtype}
\usepackage{graphicx}
\usepackage{subfigure}
\usepackage{balance}
\usepackage{booktabs} 
\usepackage{amsthm}
\usepackage{array}
\usepackage{graphicx}
\usepackage{clrscode}
\usepackage{subfigure}
\usepackage{multirow}
\usepackage{multicol}
\usepackage{float}
\usepackage{color}
\usepackage{xcolor}
\usepackage{amsopn}
\usepackage{mathrsfs}
\usepackage{mathtools}
\usepackage{amsmath}
\usepackage{booktabs}
\usepackage{arydshln}
\usepackage{hyperref}
\usepackage{blkarray}
\usepackage{enumerate}
\usepackage{courier}
\usepackage{mathrsfs}
\usepackage{rotating}
\usepackage{bm}
\usepackage{wrapfig}
\usepackage{subfigure}
\usepackage{array}
\usepackage{ragged2e}
\usepackage{amsmath}
\usepackage{threeparttable}

\theoremstyle{definition}
\newtheorem{definition}{Definition}
\theoremstyle{theorem}
\newtheorem{mytheory}{Theorem}

\theoremstyle{proof}

\theoremstyle{remark}
\newtheorem*{remark}{Remark}

\usepackage[lined,boxruled,commentsnumbered]{algorithm2e}

\begin{document}
\title{Debiased Recommendation with User Feature Balancing}
\author{Mengyue Yang~}
 \affiliation{University College London} 
 \email{mengyue.yang.20@ucl.ac.uk}
\author{Guohao Cai~}
 \affiliation{Noah's Ark Lab, Huawei} 
 \email{caiguohao1@huawei.com}
\author{Furui Liu~} 
 \affiliation{Noah's Ark Lab, Huawei}
 \email{liufurui2@huawe.com}
\author{Zhenhua Dong~} 
 \affiliation{Noah's Ark Lab, Huawei}
 \email{dongzhenhua@huawei.com  }
\author{Xiuqiang He~} 
 \affiliation{Noah's Ark Lab, Huawei}
 \email{hexiuqiang1@huawei.com}
\author{Jianye Hao~} 
 \affiliation{Noah's Ark Lab, Huawei}
 \email{haojianye@huawei.com}
\author{Jun Wang}
 \affiliation{University College London} 
 \email{j.wang@cs.ucl.ac.uk}
\author{Xu Chen$^{*}$~}\thanks{$*$ Corresponding author}
 \affiliation{Beijing Key Laboratory of Big Data Management and Analysis Methods, Gaoling School of Artificial Intelligence, Renmin University of China}
 \email{successcx@gmail.com}

\begin{abstract}
Debiased recommendation has recently attracted increasing attention from both industry and academic communities.
Traditional models mostly rely on the inverse propensity score (IPS), which can be hard to estimate and may suffer from the high variance issue.
To alleviate these problems, in this paper, we propose a novel debiased recommendation framework based on user feature balancing.
The general idea is to introduce a projection function to adjust user feature distributions, such that the ideal unbiased learning objective can be upper bounded by a solvable objective purely based on the offline dataset.
In the upper bound, the projected user distributions are expected to be equal given different items.
From the causal inference perspective, this requirement aims to remove the causal relation from the user to the item, which enables us to achieve unbiased recommendation, bypassing the computation of IPS.
In order to efficiently balance the user distributions upon each item pair, we propose three strategies, including clipping, sampling and adversarial learning to improve the training process.
For more robust optimization, we deploy an explicit model to capture the potential latent confounders in recommendation systems.
To the best of our knowledge, this paper is the first work on debiased recommendation based on confounder balancing.
In the experiments, we compare our framework with many state-of-the-art methods based on synthetic, semi-synthetic and real-world datasets.
Extensive experiments demonstrate that our model is effective in promoting the recommendation performance.
\end{abstract}
\maketitle

\section{Introduction}
Recommender system, as an effective remedy for information overloading, is playing a key role in the modern E-commerce.
Traditional recommender models are directly learned based on the observed user behaviors, which may have been contaminated by the exposure or selection bias~\cite{DBLP:conf/icml/SchnabelSSCJ16}.
Typically, the users in a recommender system can only access and evaluate a small part of the whole items due to the effect of the former recommendation algorithm or the user intrinsic tendency. 
This means that the observed data is a skewed version of the real user preference, which may bias the recommender models and limit the performance.

For solving this problem, the most reliable method is conducting online random experiments.
However, in order to collect enough training samples, such random policy has to occupy the business traffic for a long time, which can be detrimental to the user experience.
As a result, recent years have witnessed many offline debiasing methods~\cite{DBLP:conf/icml/SchnabelSSCJ16,Non-display, chen2020bias, zhang2021causal, wei2021model, wang2021clicks}, among which the inverse propensity score (IPS) is one of the most popular strategies.
The basic idea of IPS is to adjust the offline data distribution to be align with the ideal learning objective by re-weighting the training samples.
While IPS-based models have been widely leveraged for debiased recommendation, there are some significant weaknesses:
to begin with, different from the classical IPS models~\cite{DBLP:journals/corr/SwaminathanJ15, DBLP:conf/nips/SwaminathanJ15, direct1}, which are mostly evaluated based on small scale or simulated datasets, in the recommendation domain, the sample weights are estimated from the noisy and extremely sparse datasets, which can be highly unreliable.
The error of the weights will be propagated to lower the final recommendation performance.
In addition, the IPS-based learning objective usually suffers from the high variance issue~\cite{DBLP:journals/corr/SwaminathanJ15, DBLP:conf/nips/SwaminathanJ15, DBLP:journals/jmlr/SwaminathanJ15, DBLP:conf/aaai/NaritaYY19}, which makes the learned model quite fragile in facing with the complex recommendation environments. 

In order to avoid these problems, in this paper, we propose a novel debiased recommendation framework based on confounder balancing.
In specific, we formulate the recommendation task by a causal inference tool called potential outcome framework (POF).
Based on POF, the ideal learning objective is induced by the causal graph presented in Figure~\ref{intro}(b), where there is no edge from the user to the item, and the items are uniformly exposed to each user.
However, in practice, the offline dataset is generared by the causal graph shown in Figure~\ref{intro}(c), where the exposure of an item is influenced by the user.
As a result, the above ideal objective is usually intractable only with the offline dataset.
To alleviate this problem, we introduce a projection function to adjust the user feature distributions, aiming to derive a solvable upper bound of the ideal learning objective. 
In this upper bound, the user feature distributions are expected to be balanced given different items.
This requirement actually aims to cut down the causal relation between the user and item for achieving debiased recommendation, where we do not need to compute IPS in the whole modeling process.

While this seems to be a promising idea, it is non-trivial due to the following challenges:
to begin with, while there are many studies on causal inference based on confounder balancing~\cite{DBLP:journals/corr/abs-2001-07426, shalit2017estimating, DBLP:conf/icml/JohanssonSS16, DBLP:conf/nips/LiuGRKFDB18, DBLP:conf/iclr/00010K21}, they mostly focus on treatment-effect estimation, where the loss function is limited to the mean square error (MSE) and the treatment is assumed to be binary.
However, in the recommendation problem, the user preference label is usually categorical implicit feedback, which should be better optimized with the cross-entropy loss.
Directly replacing the loss function in previous work is not easy, since they are highly tailored for MSE.
More principled methods are needed to derive the learning objective suitable for the recommendation task.
In addition, the large number of items in recommender systems makes the binary treatment assumption unpractical.
More effective and efficient objectives are needed to handle the large treatment space.
At last, there can be many unobserved factors (or called latent confounders) which simultaneously influence the items and user feedback, for example, the sales promotion or user emotions. 
Such factors make it hard to cut down the user-item causal relation only by adjusting user feature distributions.
How to model them is still unexplored in the recommendation domain.

For solving these challenges, we derive a general upper bound of the ideal learning objective based on Jasen's gap bound~\cite{gao2017bounds}, which is compatible with any loss function.
In this upper bound, the user distributions are expected to be equal given different items.
Since the number of items can be very large in practice, we design three strategies including clipping, sampling and adversarial learning to facilitate model optimization.
In order to handle the latent confounders, we deploy a neural model to infer them, which are then integrated into our framework to make more robust user feedback estimation.
Our main contributions can be summarized as follows:

(\romannumeral1) we propose to build debiased recommender models based on confounder balancing, which, to the best of our knowledge, is the first time in the recommendation domain.

(\romannumeral2) To achieve our idea, we derive a surrogate learning objective, which can be optimized based on the offline dataset.
To make this objective more effective, we propose three strategies to handle the potential large number of items, and deploy an explicit model to capture the latent confounders.

(\romannumeral3) We conduct extensive experiments based on different types of datasets to demonstrate the effectiveness of our framework.

\begin{figure}[t]
\centering
\setlength{\fboxrule}{0.pt}
\setlength{\fboxsep}{0.pt}
\fbox{
\includegraphics[width=.65\linewidth]{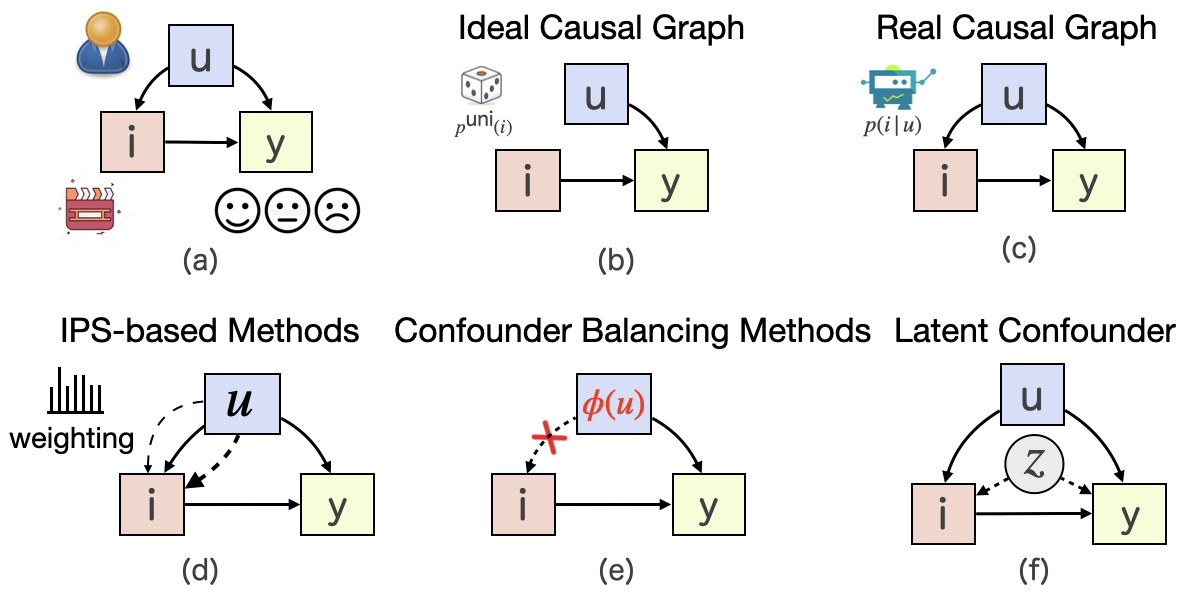}
}
\caption{
(a) Illustration of the recommendation task within the potential outcome framework.
(b) The causal graph of the ideal unbiased learning objective.
(c) The causal graph generating the offline dataset.
(d)-(e) Comparison on the debiasing mechanisms between the IPS-based and confounder balancing methods.
(f) Illustration of the latent confounders in potential outcome framework
}
\label{intro}
\end{figure}

\section{Preliminaries}
\subsection{Potential Outcome Framework}
Potential outcome framework (POF) is an effective tool for causal estimation and discovery under both experimental and observational settings.
In POF, there are three basic elements:
\textbf{unit (U)} is the atom research objective, e.g., a user in the recommender system.
\textbf{Treatment (I)} is the action imposed on the unit, e.g., an item recommended to the user.
\textbf{Outcome (Y)} is the result after applying the treatment to the unit, e.g., the user feedback on the item. 
The relations between different variables in POF is described by a \textbf{causal graph}.
For example, the recommendation task can be formulated by the causal graph in Figure~\ref{intro}(a), where the user is a common cause of the item and user feedback, and the user feedback is jointly determined by the user and item.
The variables simultaneously influence the treatment and outcome is called \textbf{confounder}.
In Figure~\ref{intro}(a), the user is a confounder of the item and user feedback.
Specially, if $p(U=u|I=i)=p(U=u|I=i')$ holds for any treatment pair $(i,i')$, then $U$ is independent with $I$, that is, there is no causal relation between $U$ and $I$.
At this time, we say the confounder is balanced on the treatment.
In the following, we indiscriminately use the notations of confounder balancing and user feature balancing.
Since POF is not the focus of this paper, we refer the readers to~\cite{yao2020survey} for more details.

\subsection{Debiased Recommendation}
Debiased recommendation is a recently emerged research topic.
Here, we explain it from the perspective of potential outcome framework.
As mentioned above, the recommendation task can be formulated by the causal graph in Figure~\ref{intro}(a).
Formally, let $u$, $i$ and $y$ denote the user, item and user feedback on the item.
Suppose $\delta$ is a loss function, $f$ is a recommender model, then the ideal learning objective for debiased recommendation is:
\begin{equation}
\label{idea}
\!\!L_{\text{ideal}} \!=\!\mathbb{E}_{u}[\frac{1}{N}\!\!\sum_{i=1}^{N} \!\delta(r_{ui},\!f(u,\!i))] \!=\! \mathbb{E}_{u}[\mathbb{E}_{i\sim p^{\text{uni}}(i)}[\delta(r_{ui},\!f(u,\!i))]],
\end{equation}
where N is the number of items. 
$r_{ui}=\mathbb{E}[y|u,i]$ is the expectation of the user feedback on the item. 
$p^{\text{uni}}(i) = \frac{1}{N}$ is the uniform distribution supported by the whole item set.
We regard the user feedback as a random variable, which can be more practical considering the potential noisy information in recommender systems.

In this objective, the training data is assumed to be generated by a causal graph (see Figure~\ref{intro}(b)), where there is no edge from the user to the item, and the item is uniformly exposed to the users. 
However, this causal graph is usually not aligned with the one (see Figure~\ref{intro}(c)) inducing the offline datasets, thus the ideal learning objective can be intractable.
In order to achieve debiased recommendation based on offline datasets, IPS is a very common strategy~\cite{DBLP:journals/jmlr/SwaminathanJ15, DBLP:conf/aaai/NaritaYY19}, where the offline samples are re-weighted to approach the data distribution induced from the ideal causal graph.
While IPS-based methods have achieved many promising results, the sample weights need to be predicted from the offline dataset, which can be inaccurate and may lead to the high variance issue~\cite{DBLP:journals/corr/SwaminathanJ15, DBLP:conf/nips/SwaminathanJ15}.
In the following sections, we propose an alternative debiased recommendation strategy based on confounder balancing\footnote{Hereafter, we do not distinguish confounder balancing with user feature balancing.}, which avoids the drawbacks of IPS.

\section{Debiased Recommendation with Confounder Balancing}
In this section, we detail our idea.
To begin with, we derive an upper bound of the ideal learning objective~(\ref{idea}) by adjusting the user distributions.
In the derived upper bound, the confounder (i.e., user) needs to be balanced for each item pair.
Since the number of items can be quite large in practice, we design three strategies including clipping, sampling and adversarial learning to facilitate the model training process.
For more robust optimization, we, at last, propose to capture the potential latent confounders in recommender systems, and combine them into the final learning objective.

\subsection{Upper Bound of the Ideal Objective}
In IPS-based methods, the ideal learning objective is achieved by re-weighting the offline training samples.
In order to avoid the weaknesses of IPS, we take a different strategy.
In specific, we hope to derive a surrogate learning objective, satisfying two properties:
on one hand, it should be strictly not smaller than~(\ref{idea}), so that minimizing the surrogate can approximately lower the ideal learning objective. 
On the other hand, the derived objective should be solvable only with the offline dataset.
To proceed, we introduce a representative function $\phi:\mathcal{U}\rightarrow\mathbb{R}^d$, which projects the user information (defined on $\mathcal{U}$) into a $d$-dimensional space. 
Then we have the following theory:
\begin{mytheory}
\label{theory1}
For measuring distribution distance, suppose $\text{IPM}_G$ is the integral probability metric (IPM) defined in~\cite{ipmmetrics}, that is: 
\begin{equation}
    IPM_G(p, q) := \sup_{g\in G}\int_S| g(s) (p(s)-q(s))ds|,
\end{equation}
where $G$ is a function class, $p$ and $q$ are two density functions on $S$.

Let $l_{f,\phi}(u, i)=\int_y\delta(y, {f}(\phi(u),i))p(y|u,i)\text{d}y$ be the local loss for the user-item pair $(u,i)$.
We define $s_i=\int_u l_{f,\phi}(u, i)p(u|i)\text{d}u$, where $p(u|i)$ is the probability of observing $u$ given $i$.
Then for any loss function $\delta$, we have:
\begin{equation}
\begin{aligned}
\label{up-bound}
&\mathbb{E}_{u} [\mathbb{E}_{i\sim p^{\text{uni}}(i)}[\delta(r_{ui},f(\phi(u),i))]]\\
\leq &\frac{1}{N}\sum_{i=1}^N s_i+ \frac{1}{N}\!\!\!\!\sum_{{i,i'\in[1,N],\atop i\not=i'}}\!\!\!\!\{p(i)\!+\!p(i')\}B_\phi \text{IPM}_G(p_\phi^i,p_\phi^{i'})\!+\!C,
\end{aligned}
\end{equation}
where 
N is the number of items.
$p(i)$ is the marginal probability of observing item $i$.
$B_\phi$ is a constant, such that $\frac{l_{f,\phi}(u, i)}{B_\phi} \in G,~~\forall i\in [1, N]$.
$p_\phi^i$ is the user distributions by projecting $p(u|i)$ with $\phi$, that is, $p_\phi^i(\phi(u)) \overset{def}{=} p_\phi(\phi(u)|i)$.
C is a constant.
\end{mytheory}

\begin{proof}
To begin with, we decompose the ideal objective~(\ref{idea}) as follows:
{\setlength\abovedisplayskip{5pt}
\setlength\belowdisplayskip{5pt}
\begin{equation}
\begin{split}\label{app-idea}
    L_{\text{ideal}} =& \mathbb{E}_{u} [\mathbb{E}_{i\sim p^{\text{uni}}(i)}[\delta(r_{ui},f(\phi(u),i))]]\\
    =&\int_{u} \frac{1}{N}\sum_{{i}} \delta(r_{ui}, {f}(\phi(u),i))p(u) \text{d}u \\
    =&\frac{1}{N}\sum_{{i}} \int_{u, j} \delta(r_{ui}, {f}(\phi(u),i))p(u, j)  \text{d}u  \text{d}j
\end{split}
\end{equation}}
where $p(u, j)$ is the probability of user-item pair $(u,j)$ in the observed dataset.
The last equation holds due to the properties on marginal distribution.
$\delta(\cdot, \cdot)$ can be arbitrary loss functions such as RSME, binary cross entropy, among others. 

To bound equation~(\ref{app-idea}), we introduce the following two notations: 
{\setlength\abovedisplayskip{5pt}
\setlength\belowdisplayskip{5pt}
\begin{equation}
\begin{aligned}
    &k_1(i) = \int_{u} \delta(r_{ui}, {f}(u, i))p(u|i) du \\
    &k_2(i, i') = \int_{u}\delta(r_{ui}, {f}(u, i))p(u|i') du \quad i\neq i'
\end{aligned}
\end{equation}}
Then the ideal learning objective can be write as:
{\setlength\abovedisplayskip{5pt}
\setlength\belowdisplayskip{5pt}
\begin{equation}
\begin{aligned}
\label{expected_objective}
L_{\text{ideal}} =\frac{1}{N}\sum_{{i}} [p(i) k_1(i) + \sum_{{i'\not=i}} p(i')k_2(i,i')]
\end{aligned}
\end{equation}}

By the above definition, we relate the idea learning objective with the observed data distribution.
Recall that $r_{ui}=\mathbb{E}[y|u,i]$, which is intractable only from empirical samples.
Thus, we further define the following notations:
{\setlength\abovedisplayskip{5pt}
\setlength\belowdisplayskip{5pt}
\begin{equation}
\begin{aligned}
\label{expected}
\epsilon_F &= \int \delta(y_{ui}, {f}(u,i)) p(y_{ui}|u,i)p(u, i) du di dr_i \\ 
&=\sum_{{i}} p(i) \underbrace{\int_u[\int_{y_{ui}}\delta(y_{ui}, {f}(u,i)p(y_{ui}|u,i)dy_{ui})] p(u|i)du}_{s_i}=\sum_{{i}} p(i)s_i\\
\epsilon_{CF}^i &= \sum_{i'\not=i}\int \delta(y_{ui}, {f}(u, i)) p(y_{ui}|u,i)p(u, i') du dy_{ui}\\
&=\sum_{i'\not=i} p(i') \underbrace{\int_u[\int_{y_{ui}}\delta(y_{ui}, {f}(u,i)p(y_{ui}|u,i)dy_{ui})] p(u|i')du}_{t_{ii'}}=\sum_{i'\not=i} p(i')t_{ii'}
\end{aligned}
\end{equation}}
where $y_{ui}$ is the observed feedback of user $u$ on item $i$\footnote{It should noted that, in the main paper, we use $u_t$, $i_t$ and $y_t$ to represent the same concepts.}.

Based on the above definition, we formally derive the upper bound of $L_{\text{ideal}}$ following three key steps:
(1) the first step is to relate $k_1(i)$ and $k_2(i, i')$ with $s_i$ and $t_{ii'}$ by bounding $|s_i-k_1(i)|$ and $|t_{ii'}-k_2(i, i')|$, respectively. Then $L_{\text{ideal}}$ can be represented by $s_i$ and $t_{ii'}$.
(2) Since $t_{ii'}$ contains unobserved user-item pairs which are intractable, the second step is to connect $s_i$ and $\epsilon_{CF}^i$, so that $L_{\text{ideal}}$ can be only depend on $s_i$.
(3) Based on (1) and (2), we derive the upper bound of $L_{\text{ideal}}$ to end the proof.

\textbf{(1) Relating $k_1(i)$ and $k_2(i, i')$ with $s_i$ and $t_{ii'}$:}
Let $P_i = |s_i-k_1(i)|$ and $Q_i = |t_{ii'}-k_2(i, i')|$, then we have following Lemma:
\begin{lemma}
\label{lem:p_q}
If the loss function $\delta$ is differentiable, for any item $i$. There exist constants $M_i>0, \beta>0, \gamma\ge\beta$, such that:
{\setlength\abovedisplayskip{5pt}
\setlength\belowdisplayskip{5pt}
\begin{equation}
\begin{aligned}
P_i &=|s_i-k_1(i)|\\
&=|\int_u E_{y_{ui}}[\delta(y_{ui}, {f}(u,i))]-\delta(r_{ui}, {f}(u,i)) p(u|i) du|\\
&\le \int_u |E_{y_{ui}}[\delta(y_{ui}, {f}(u,i))]-\delta(r_{ui}, {f}(u,i))| p(u|i) du\\
&\le \int_u M_i(\rho_{\beta i}+\rho_{\gamma i})p(u|i) du\\
&= M_i(\rho_{\beta i}+\rho_{\gamma i})\\
Q_i &=|t_{ii'}-k_2(i, i')|\\
&=|\int_u E_{y_{ui}}[\delta(y_{ui}, {f}(u,i))]-\delta(r_{ui}, {f}(u,i)) p(u|i') du|\\
&\le \int_u |E_{y_{ui}}[\delta(y_{ui}, {f}(u,i))]-\delta(r_{ui}, {f}(u,i))| p(u|i') du\\
&\le \int_u M_i(\rho_{\beta i}+\rho_{\gamma i})p(u|i') du\\
&= M_i(\rho_{\beta i}+\rho_{\gamma i})
\end{aligned}
\end{equation}}
\end{lemma}
The second inequality in $P_i$ and $Q_i$ holds because of the Jensen bound demonstrated in~\cite{gao2017bounds}.

Based on Lemma~\ref{lem:p_q}, we could get the following upper bound of $L_{\text{ideal}}$:
{\setlength\abovedisplayskip{5pt}
\setlength\belowdisplayskip{5pt}
\begin{equation}
\begin{aligned}
\label{eq:a_f_cf}
L_{\text{ideal}} =&  \frac{1}{N}\sum_{{i}} [p(i) k_1(i) + \sum_{{i'\not=i}} p(i')k_2(i,i')]\\
\le& \frac{1}{N}\sum_{{i}} [p(i) [s_i+M_i(\rho_{\beta i}+\rho_{\gamma i})] + \sum_{{i'\not=i}} p(i')[t_{ii'}+M_i(\rho_{\beta i}+\rho_{\gamma i})]]\\
=& \frac{1}{N}\epsilon_{F} + \frac{1}{N}\sum_{{i}}\epsilon_{CF}^i + \frac{1}{N}\underbrace{\sum_{{i}}[p(i)x(i)+\sum_{{i'}}p(i')x(i)]}_{constant} \\
=& \frac{1}{N}\sum_{{i}} p(i)s_i + \frac{1}{N}\sum_{{i}}\epsilon_{CF}^i + C
\end{aligned}
\end{equation}}
where $C$ is a constant.

\textbf{(2) Relating $\epsilon_{CF}^i$ with $s_i$:}
In Eq. \ref{eq:a_f_cf}, $\epsilon_{CF}^i$ cannot be estimated from the observation data, since it contains user-item pairs which are unobserved.
To solve this problem, we connect $\epsilon_{CF}^i$ with $\epsilon_{F}$ by the following lemma:
\begin{lemma}
\label{lem:f_cf}
Let $\psi$ denotes the inverse function of $\phi$, $t$ denotes representation of $u$ (i.e. $t = \phi(u)$), then:
\begin{equation}
\begin{aligned}
\label{eq:f_cf}
&\epsilon_{CF}^i - \sum_{i'\not=i}p(i')s_i\\
=&\sum_{i'\not=i}p(i')\int \mathbb{E}_{y_{ui}}[l(y_{ui}, f(u, i))] (p(u|i)-p(u|i'))dx\\
=&\sum_{i\not=i}p(i')\int \delta(\psi(t),i)(p_\phi^i(t)-p_\phi^{i'}(t)) dt\\
\le& \sum_{i'\not=i}p(i') B_\phi IPM_G(p_\phi^i(t),p_\phi^{i'}(t))
\end{aligned}
\end{equation}
\end{lemma}
In Lemma~\ref{lem:f_cf}, we use $s_i$ to bound the crucial intractable term $\epsilon_{CF}^i$. 
The $p_\phi^i(t),p_\phi^{i'}(t)$ in bottom line are the converted user feature representation for item $i$ and $i'$, respectively.
The distance between $\epsilon_{CF}^i$ and $\sum_{i'\not=i}p(i')s_i$ will be minimized when $p_\phi^i(t)$ is equal to $p_\phi^{i'}(t)$ for any $i'\neq i$.

\textbf{(3) Deriving the final upper bound of $L_{\text{ideal}}$:}
By bringing equation~(\ref{eq:f_cf}) into (\ref{eq:a_f_cf}), we have:
\begin{equation}
\begin{aligned}
L_{\text{ideal}} &= \frac{1}{N}\sum_{i=1}^{N} p(i) k_1(i) + \frac{1}{N}\sum_{i=1}^{N}\sum_{{i'\not=i}} p(i')k_2(i,i')\\
&\leq \frac{1}{N}\sum_{i=1}^{N} p(i) s_i + \frac{1}{N}\sum_{i=1}^{N}\sum_{{i'\not=i}} p(i')t_{ii'} + C\\
&\leq \frac{1}{N}\sum_{i=1}^{N} p(i) s_i + \frac{1}{N}\sum_{i=1}^{N}\sum_{{i'\not=i}} p(i')[s_i+B_\phi \text{IPM}_G(p_\phi^i,p_\phi^{i'})] + C\\
&\leq \frac{1}{N}\sum_{i=1}^{N} s_i + \frac{1}{N}\sum_{i=1}^{N}\sum_{{i'\not=i}} p(i')B_\phi \text{IPM}_G(p_\phi^i,p_\phi^{i'}) + C\\
&= \frac{1}{N}\sum_{i=1}^{N} s_i + \frac{1}{N}\sum_{{i,i'\in[1,N],\atop i\not=i'}}\!\!\!\! \{p(i')+p(i)\}B_\phi \text{IPM}_G(p_\phi^i,p_\phi^{i'}) + C\\
\end{aligned}
\end{equation}
where $C$ absorbs all the constants, and the last equation holds because $\text{IPM}_G(p_\phi^i,p_\phi^{i'})$ is symmetrical for $i$ and $i'$.
\end{proof}

On this theory, we have the following remarks:

$\bullet$ This objective mainly includes two parts: the first one aims to learn the recommender model based on the projected user information $\phi(u)$, and the second one targets at minimizing the distance between the projected user distributions given each pair of items.
Since $s_i$ is defined on the observational data, this upper bound can be directly optimized with offline datasets.

$\bullet$ The representative function $\phi$ has two effects: 
(\romannumeral1) ``adjusting'' all the irregular user feature distributions to a target one,
and (\romannumeral2) ``learning'' the recommender model based on the target distribution.
Such ``adjusting-learning'' effects actually aim to solve the key problem of distribution shift in debiased recommendation~\cite{DBLP:conf/kdd/ZouKCC019}.
For a shifted user feature distribution in the testing set, the learned $\phi$ can effectively project it to the target distribution, which has been well handled in the training phrase.

$\bullet$ From the causal graph perspective, the effect of the second part in the upper bound is to cut down the causal relation from the user to the item.
The optimal solution of minimizing $\text{IPM}_G(p_\phi^i,p_\phi^{i'})$ for each item pair is $p_{\phi}(\phi(u)|i=1)=p_{\phi}(\phi(u)|i=2)=...=p_{\phi}(\phi(u)|i=N)=p_{\phi} (\phi(u))$, which implies that $\phi(u)$ is independent of $i$, that is, there is no ``user $\rightarrow$ item'' edge in the causal graph.

$\bullet$ Notably, many previous work on confounder balancing leverage similar techniques to derive solvable learning objectives~\cite{DBLP:journals/corr/abs-2001-07426}.
However, our result is more general in two aspects: 
(\romannumeral1) previous studies are highly tailored for MSE, but our learning objective is compatible with any loss function, and (\romannumeral2) the treatment is no longer restricted to be binary, its size can be much larger.

\textbf{Empirical loss of the upper bound.}
Suppose the offline recommender dataset is $\{(u_t,i_t,y_t)\}_{t=1}^T$, where each sample $(u_t,i_t,y_t)$ includes a user, an item and the user feedback on the item, and there are in total $T$ samples. 
Then the empirical loss corresponding to the upper bound in equation~(\ref{up-bound}) is:
\begin{equation}
\begin{aligned}
\label{emp}
&L_{\text{emp}}(\theta_f,\theta_{\phi}) \\
= & \frac{1}{T}\sum_{t=1}^T\frac{1}{p(i_t)}\delta(y_t,f(\phi(u_t),i_t)) + \lambda_{f}||\theta_f||_2^2 + \lambda_{\phi}||\theta_{\phi}||_2^2+\\
&\gamma\!\!\!\!\!\sum_{{i,i'\in[1,N],\atop i\not=i'}}\!\!\!\!\{p(i)\!+\!p(i')\}\text{IPM}_G(\{\phi(u_t)\}_{t:i_t=i},\{\phi(u_t)\}_{t:i_t=i'})\!\\
\end{aligned}
\end{equation}
where $\theta_f$ and $\theta_{\phi}$ are the parameters of the recommender model $f$ and representative function $\phi$.
$p(i)$ is approximated as $\frac{T_i}{T}$, and $T_i$ is the frequency of item $i$ in the dataset.
Thus, the weight $\frac{1}{p(i_t)}$ compensates for different item observation frequencies.
Similar to~\cite{shalit2017estimating}, $\text{IPM}_G(\cdot; \cdot)$ is the empirical integral probability metric w.r.t. $G$, and $B_\phi$ is regarded as a part of the hyper-parameter $\gamma$.

\subsection{Efficient Confounder Balancing}\label{ECB}
To achieve confounder balancing, the user distribution distance is minimized for each item pair in objective~(\ref{emp}).
Suppose there are N items in the system, then we need to optimize $C_{N}^2$ balancing terms (i.e., $\text{IPM}_G(\cdot; \cdot)$), which can be quite inefficient or even infeasible in practical recommender systems.
To alleviate this problem, we propose the following three strategies to facilitate model training:

\textbf{Clipping.}
This method directly drops the balancing terms which are not important.
In specific, let $p_{ii'} = p(i')+p(i)$ be the importance score of item pair $(i,i')$.
Then we define $S$ as the set of item pairs, which has the largest $p_{ii'}$'s, and we suppose $|S|=K_1$.
At last, we change the last part of (\ref{emp}) to:
\begin{equation}
\begin{aligned}
\label{emp-clip}
\gamma\sum_{{i,i'\in S}}p_{ii'}\text{IPM}_G(\{\phi(u_t)\}_{t:i_t=i},\{\phi(u_t)\}_{t:i_t=i'}),
\end{aligned}
\end{equation}
where we directly remove the balancing terms with smaller importance scores.

\textbf{Sampling.}
In this method, we randomly sample $K_2$ balancing terms for optimization in each training epoch.
The balancing term for $(i,i')$ is sampled according to the probability $\frac{p_{ii'}}{\sum_{j,j'\in[1,N],j\not=j'} p_{jj'}}$.
By this method, the number of optimized balancing terms is reduced from $C_{N}^2$ to $K_2$, which improves the training efficiency. 
Since the balancing terms are selected in a random manner, the item pairs with smaller $p_{ii'}$'s still have chances to be optimized, which improves the model robustness as compared with the clipping strategy.

\textbf{Adversarial learning.}
While the above two methods can enhance the training efficiency, they have to abandon some balancing terms either deterministically or stochastically.
In order to fully satisfy the balancing requirements, we propose an adversarial training strategy.
Our general idea is to introduce a discriminator $D$, which serves as a classifier to identify which distribution (of $p(u|i=t), t\in[1, N]$) an instance is sampled from.
Then the representative function $\phi$ is optimized to make the instance $\phi(u)$ unidentifiable.
Formally, we have the following learning objective:
{\setlength\abovedisplayskip{6pt}
\setlength\belowdisplayskip{6pt}
\begin{equation}
\begin{aligned}
\label{emp-adv}
&L^{\text{adv}}_{\text{emp}}(\theta_f,\theta_{\phi},\theta_{D}) \\
= & \frac{1}{T}\sum_{t=1}^T\frac{1}{p(i_t)}\delta(y_t,f(\phi(u_t),i_t)) + \gamma\sum_{i=1}^N \mathbb{E}_{u\sim p(u|i)}[\log D^i(\phi(u))]\\
+ &\sum_{k\in\{f,\phi,D\}}\lambda_{k}||\theta_k||_2^2,\\
\end{aligned}
\end{equation}}
where $\theta_{D}$ is the parameter of $D$.
The output layer of $D$ is softmax for classification, and $D^i$ denotes the $i$th element. 
Thus, $\sum_{i=1}^N D^i(\phi(u)) = 1$.
$\mathbb{E}_{u\sim p(u|i)}$ can be empirically estimated by $\sum_{u\in I_u}$, and $I_u$ is the item set interacted by user $u$.
We optimize this objective based on the following constraint minimax game:
{\setlength\abovedisplayskip{6pt}
\setlength\belowdisplayskip{6pt}
\begin{equation}
\begin{aligned}
\label{emp-adv1}
\min_{\theta_{f}, \theta_{\phi}}\max_{\theta_{D}} &~L^{\text{adv}}_{\text{emp}}(\theta_f,\theta_{\phi},\theta_{D}) \\
s.t. &\sum_{i=1}^N D^i(\phi(u)) = 1,
\end{aligned}
\end{equation}}
where $\theta_{D}$ is learned to maximize the likelihood of the right classification results.
$\theta_{f}$ is optimized to fit the recommendation data.
$\theta_{\phi}$ is learned to jointly hide the identification of $\phi(u)$ and achieve better recommendation results.
To see the effect of the above adversarial learning strategy, we have the following theory:
\begin{mytheory}
\label{theory2}
Let $p_\phi^i$ denotes the user distributions by projecting $p(u|i)$ with $\phi$, then the minimax game defined by
{\setlength\abovedisplayskip{3pt}
\setlength\belowdisplayskip{3pt}
\begin{equation}
\begin{aligned}
\label{emp-adv2}
\min_{\theta_{\phi}}\max_{\theta_{D}}& \sum_{i=1}^N \mathbb{E}_{u\sim p(u|i)}[\log D^i(\phi(u))] \\
s.t. &\sum_{i=1}^N D^i(\phi(u)) = 1
\end{aligned}
\end{equation}}
has a global optimal solution when $p_\phi^1 = p_\phi^2=...=p_\phi^N$.
\end{mytheory}

\begin{proof}
Suppose $t = \phi(u)$, then objective~(\ref{emp-adv2}) can be written as:
\begin{equation}
\begin{aligned}
\max_{\theta_{D}}&\sum_{i=1}^N \int_t p_\phi^i(t) \log D^i(t) dt \\
s.t. &\sum_{i=1}^N D^i(t) = 1
\end{aligned}
\end{equation}
We use Lagrange multiplier to solve the above constraint optimization problem, which induces the following objective:
\begin{equation}
\begin{aligned}
L(D, \lambda)=\max_{\theta_{D}}& \sum_{i=1}^N \int_s p_\phi^i(t) \log D^i(t)dt +\lambda (\sum_{i=1}^N D^i(t)-1)\\
\end{aligned}
\end{equation}

Let $\frac{\partial L(D, \lambda)}{\partial D}=0$,$\frac{\partial L(D, \lambda)}{\partial \lambda}=0$, we have:
\begin{equation}\label{eq:partical}
\begin{aligned}
    D^i(t) = -\frac{p_\phi^i(t)}{\lambda}\\
    \lambda = -\sum_k^N p_\phi^k(t)
\end{aligned}
\end{equation}
We then substitute Eq. \ref{eq:partical} into $\sum_{i=1}^N \int_t p_\phi^i(t) \log D^i(t)$, and obtain:
\begin{equation}
\begin{aligned}
    B = \sum_{i=1}^N \int_s p_\phi^i(s) \log \frac{p_\phi^i(t)}{\sum_k^N p_\phi^k(t)}dt
\end{aligned}
\end{equation}

The following proof is based on Jensen-Shannon Divergence (JSD), which is defined as follows:
\begin{definition}(JSD~\cite{fuglede2004jensen}) Denote $P_i$ a random distribution where $i\in\{1,2,\cdots,N\}$. Let $M:=\sum_{i=1}^{N} \frac{1}{N} P_{i}$, the general multivariate ensen-Shannon Divergence (JSD) is defined as
$$
\begin{aligned}
\mathrm{JSD}\left(P_{1}, P_{2}, \ldots, P_{n}\right) &=\sum_{i} \frac{1}{N} KL\left(P_{i} \| M\right) \\
&=H\left(\sum_{i=1}^{n} \frac{1}{N} P_{i}\right)-\sum_{i=1}^{n} \frac{1}{N} H\left(P_{i}\right)
\end{aligned}
$$
where KL is the Kullback–Leibler divergence, and H is the Shannon entropy. According to~\cite{fuglede2004jensen}, JSD is maximized when $P_{1}=P_{2}=...=P_{N}$.
\end{definition}

Finally, we have:
\begin{equation}
\begin{aligned}
    B &= \sum_{i=1}^N \int_s p_\phi^i(s) \log \frac{p_\phi^i(t)}{\sum_k^N p_\phi^k(t)}dt\\
    &=\sum_{i=1}^N E_{p_\phi^i(s)}[\log \frac{p_\phi^i(t)}{\sum_k^N p_\phi^k(t)}]+\log N\\
    &=\sum_{i=1}^N D_{KL}(p_\phi^i(t)||\frac{1}{N}\sum_k^N p_\phi^k(t))+\log N\\
    &=N \text{JSD}(p_\phi^1(t), p_\phi^2(t), \cdots, p_\phi^N(t))
\end{aligned}
\end{equation}
The above equation is maximized when all $p_\phi^i(t)$'s are equal, that is:
\begin{equation}
\begin{aligned}
    p_\phi^1(t) = p_\phi^2(t) = \cdots = p_\phi^N(t)
\end{aligned}
\end{equation}
Thus, the optimal solution for the adversarial objective in equation~(\ref{emp-adv}) is the same as that of the IPM balancing terms in equation~(\ref{emp}). 
\end{proof}
From this theory, we can see, the optimal solution is exactly what we expect from the pair-wise balancing terms in objective~(\ref{emp}).
However, by this adversarial learning strategy, the number of optimization terms is reduced from $\mathcal{O}(N^2)$ to $\mathcal{O}(N)$.

\subsection{Latent Confounder Modeling}
Real-world recommender systems are usually very complex, there can be many unobserved factors (e.g., item promotion, user emotions) that simultaneously influence the item exposure and user feedback. Such factors are called latent confounders in potential outcome framework (see Figure~\ref{intro}(f)).
The existence of latent confounders may: 
(\romannumeral1) lead to lowered recommendation performance due to the lack of sufficient context modeling,
and (\romannumeral2) introduce uncontrollable cause to influence the item exposure probability, which cannot be removed by just balancing the user distributions, and thus 
influences the debiasing effect.
In order to capture latent confounders, we follow the previous work~\cite{wang2019blessings,ranganath2018multiple,guo2020learning} to use a neural network $c$ to infer them.
More specifically, it is intuitive that latent confounders should be related with user personalities, for example, the confounders on user emotions.
In addition, according to the causal graph in Figure~\ref{intro}(f), the latent confounders are also related with the item.
As a result, we estimate the latent confounders by: $z_t = c(i_t, u_t)$.
It should be noted that the latent confounder prediction is an inference process, it does not mean there is a causal edge from the item to the latent confounders, since causal graph describes the data generation process.
Besides, the output $y_t$ is not leveraged as a signal to estimate $z_t$, since it is not observed in advance.
Similar to the previous work~\cite{wang2019blessings,ranganath2018multiple,guo2020learning}, we learn the parameters in $c$ by maximizing the likelihood of observing the exposed items.
The final learning objective by incorporating the latent confounders is:
{\setlength\abovedisplayskip{10pt}
\setlength\belowdisplayskip{10pt}
\begin{equation}
\begin{aligned}
\label{final}
&L^{\text{conf}}_{\text{emp}}(\theta_f,\theta_{\phi},\theta_{D},\theta_{c},\theta_{s}) \\
= & \frac{1}{T}\!\sum_{t=1}^T\!\frac{1}{p(i_t)}\delta(y_t,\!f(\phi(u_t, z_t),i_t)) \!+\!\gamma\!\sum_{i=1}^N \mathbb{E}_{u\sim p(u|i)}[\log D^i(\phi(u, z_t))]\\
- & \frac{1}{T}\sum_{t=1}^T \log p_s(i_t|u_t, z_t) + \!\!\!\sum_{k\in\{f,\phi,D,c,s\}}\!\!\!\lambda_{k}||\theta_k||_2^2 \\
\end{aligned}
\end{equation}}
where $\theta_{c}$ collects all the parameters in $c$. The representative function $\phi$ is applied to both $u_t$ and $z_t$, the third term aims to maximize the likelihood of observing item $i_t$ under the effect of the latent confounders.
$p_{s}$ is a probability function with $\theta_{s}$ as its parameters.

\setlength{\textfloatsep}{0.1cm}
\begin{algorithm}[t] 
\caption{Learning algorithm of our model} 
\label{alg} 
Randomly initialize $\theta_f,\theta_{\phi},\theta_{D},\theta_{c},\theta_{s}$.\\
Indicate the number of training epochs $E$.\\
Indicate the iteration numbers $E_D$ and $E_G$.\\
Indicate the learning rates $\alpha$ and $\beta$.\\
\For{e in [0, $E$]}
{
    $\hat{\theta}_c\leftarrow {\theta_c}$, $\hat{\theta}_{\phi}\leftarrow \theta_{\phi}$\\
    \For{k in [0, $E_D$]}{
        \#~~$L_{dis}(\theta_D, \hat{\theta}_c, \hat{\theta}_{\phi})\!=\!\sum_{i=1}^N \mathbb{E}_{u\sim p(u|i)}[\log D^i(\phi(u, z_t))]$.\\
        $\theta_D\leftarrow \theta_D + \alpha \frac{\partial L_{dis}(\theta_D, \hat{\theta}_c, \hat{\theta}_{\phi})}{\partial \theta_D}$.\\
    }
    $\hat{\theta}_D\leftarrow {\theta_D}$.\\
    \For{l in [0, $E_G$]}{
        \For{k in \{~f,~$\phi$,~c,~s~\}}{
           $\theta_k\leftarrow \theta_k - \beta \frac{\partial L^{\text{conf}}_{\text{emp}}(\theta_f,\theta_{\phi},\hat{\theta}_D,\theta_{c},\theta_{s})}{\partial \theta_k}$.\\
         }
    }
}
\end{algorithm}

\textbf{Model specification.}
In the above sections, we have detailed our idea.
We specify different parts of objective~(\ref{final}) as follows:
$\delta$ is implemented by the cross-entropy loss to model user implicit feedback.
$\phi$ is a simple linear function.
The categorical features for the users are firstly converted to one-hot vectors, and then encoded by embedding matrices.
The continuous features are directly multiplied by weighting matrices to derive the representation.
The discriminator $D$ is a fully connected neural network, that is, $D(\phi(u_t, z_t)) = \sigma(W_1\text{ReLU}(W_2\phi(u_t, z_t)))$,
where $W_1$ and $W_2$ are weighting parameters, $\sigma$ and ReLU are activation functions. 
The model for inferring latent confounders $c$ is a two layer neural network with ReLU as the activation function, that is, $c(u_t;i_t) = V_1\text{ReLU}(V_2\text{ReLU}(V_3[u_t;i_t]))$, where $V_1$, $V_2$ and $V_3$ are weighting parameters.$[\cdot;\cdot]$ is the concatenate operation.
$p_s$ is also a two layer neural network, but the input is $u_t$, $z_t$ and $i_t$, and the output is a scalar between 0 and 1, that is, $p_s(i_t|u_t, z_t) = \sigma(Q_1\text{ReLU}(Q_2\text{ReLU}(Q_3[u_t;i_t;z_t])))$,where $Q_1$, $Q_2$ and $Q_3$ are weighting parameters.

\textbf{Learning algorithm.}
The complete learning algorithm of our framework is summarized in Algorithm~\ref{alg}.
To begin with, the model parameters are randomly initialized.
In each training epoch, the discriminator $\theta_D$ is firstly optimized by maximizing $L_{dis}$ with fixed $\hat{\theta}_{g}$ and $\hat{\theta}_{\phi}$.
Then the recommender model, representative function and likelihood of the treatment are jointly optimized to minimize $L^{\text{conf}}_{\text{emp}}$ by fixing the discriminator $\hat{\theta}_{D}$.

\begin{remark}
In IPS-based methods, the data distribution is adjusted to the ideal one by re-weighting the training samples, which basically smooths the conditional probability corresponding the causal relation ``$u\!\rightarrow\!i$'' (see Figure~\ref{intro}(d)). 
In our user feature balancing method, we project the user distribution to make it independent of the item, which cuts down the causal relation ``$u\!\rightarrow\!i$'' (see Figure~\ref{intro}(e)).
Basically, IPS and our method achieve debiased recommendation with totally different principles.
In the previous work, IPS-based methods have been well studied. 
In this paper, we propose the first solution to debiased recommendation based on confounder balancing, which paves a new way for this research field.
\end{remark}

\section{Related Work}
Our work aims to push the boundary of debiased recommendation, which holds the promise of learning user unbiased preference for adapting complex testing environments.
Traditionally, there are two types of strategies to build debiased recommender models \cite{zhao2021popularity, ding2019reinforced, chen2019counterfactual, zhao2021popularity, wang2021probabilistic, ding2021causal}.
The first one is sample generation methods, where the key idea is to impute user preference on unexposed items, so that the model can be directly trained on the full sample space. 
For example, \cite{direct1} proposes to train an imputation model based on the observed data, and then the recommendation algorithm is learned based on both of the imputated and original samples.
The second strategy is distribution adjusting methods~\cite{saito2020unbiased,DBLP:conf/icml/SchnabelSSCJ16}.
The main idea is to re-weight the training samples for adjusting the offline data distribution to be similar with the ideal unbiased one.
Almost all the models in this category rely on IPS, for example, 
\cite{DBLP:journals/corr/SwaminathanJ15, DBLP:conf/nips/SwaminathanJ15, direct1,chen2021autodebias} leverage different strategies to estimate the inverse propensity scores, which are used to re-weight the user-item interactions.
\cite{DBLP:conf/nips/SwaminathanJ15} proposes to compute the propensity score by self-normalization to reduce the variance.
As mentioned above, IPS-based methods can be flawed when facing with the recommendation domain, since the sample weights need to be estimated from the noisy recommendation data, which can be unreliable and lower the downstream recommendation performance.
There are also many studies on combining the former two methods, aiming to take both of their advantages.
For example, \cite{DBLP:conf/icml/DudikLL11} proposes to simultaneously estimate IPS and learn imputation models to obtain more robust performance. 
\citet{Non-display} introduces uniform data to train the imputation model.
Our idea lies in the second strategy.
However, we do not estimate the sample weights, and thus do not suffer from the weaknesses, such as low estimation accuracy and high variance, brought by IPS-based models.

In another related field of causal inference, there are many studies on confounder balancing.
For example, 
~\cite{johansson2016learning} estimates the average treatment effect (ATE) by balancing the heterogeneous confounders.
~\cite{shalit2017estimating} uses representation learning to balance the confounders to predict the individual treatment effect (ITE).
~\cite{ma2021deconfounding} also aims to learn ATE, but it targets at networked data, where the confounder is balanced by considering the graph structure information.
These work mostly aim to solve the treatment-effect estimation problem, which is quite different from the recommendation task, ranging from the loss function, treatment space to the potential latent confounders. 
Recently, we have noticed that~\citet{wang2021deconfounded} introduces a confounder inference method for building debiased recommender models.
The basic idea of this work comes from Pearl's causal inference framework\cite{pearl2009causality}, which focus on the identifiability problem based on the back-door or front-door criterion. 
In our work, we build debiased recommender models based on potential outcome framework, and we achieve this goal by balancing user feature distributions for different items.

\section{Experiments}
In this section, we conduct extensive experiments to demonstrate the effectiveness of our idea, focusing on the following research questions:

(1) what is the overall performance of our model comparing with the state-of-the-art methods?

(2) What are the effects of different model components?

(3) How the data bias severity influence the recommendation performance?

(4) How the data confounderness influence the recommendation performance?

(5) How the data sparsity influence the recommendation performance?

In the following sections, we firstly detail the experiment setup, and then answer these questions based on the experiment results.
\subsection{Experiment Setup}
\subsubsection{Datasets}
We base our experiments on both synthetic and real-world datasets.
The synthetic dataset is built according to the method leveraged in~\cite{DBLP:conf/kdd/ZouKCC019}.
In specific, we simulate 10000 users and 32 items.
For each user $i$ or item $j$, the features $\bm{p}_i\in \mathbb{R}^{d}$ or $\bm{q}_j\in \mathbb{R}^{d}$ are generated from a multi-variable Gaussian distribution $\mathcal{N}(\bm{0},\bm{I})$, where $d$ and $\bm{I}$ represent the feature dimension and unit matrix, respectively.
When generating the recommendation list for user $i$, we introduce a random variable $z\sim \mathcal{N}(0, 1)$ to model the influence of the latent confounders.
In specific, for each item $j$, we define a score $r_{ij}= 1-\sigma[(1-\alpha)(1-\beta)\bm{a}^T\kappa_1(\kappa_2([\bm{p}_{i}, \bm{q}_{j}])) + \alpha + \beta z+\mathcal{N}(0, 0.02))]$, 
where $\bm{a}\in \mathbb{R}^{2d}$ is specified as an all-one vector.
$\alpha\in[0, 1]$ defines the bias severity, e.g., $\alpha=1$ means the item is recommended in a complete random manner, that is, the recommendation of an item is irrelevant with the user.
$\beta\in [0, 1]$ trades-off the influence of latent confounders, e.g., $\beta=1$ implies that the latent confounders dominate the system.
In the experiment, the default settings of $\alpha$ and $\beta$ are both 0.5.
$\kappa_1(\cdot)$ and $\kappa_2(\cdot)$ as piecewise functions~\cite{DBLP:conf/kdd/ZouKCC019}.
$\kappa_1(x) = x - 0.5$ if $x>0$, otherwise $\kappa_1(x) = 0$.
$\kappa_2(x) = x$ if $x>0$, otherwise $\kappa_2(x) = 0$.
For each user $i$, we recommend item $j$ according to a Bernoulli distribution with the mean of $r_{ij}$.
The length of the recommendation list is set as 5. 
We follow the previous work~\cite{DBLP:conf/kdd/ZouKCC019} to generate the feedback of a user $i$ on an item $j$ as follows:
\begin{equation}
\begin{split}
    &x_{1} =  \bm{a}^T\kappa_1(\kappa_2([\bm{p}_{i}, \bm{q}_{j}]))+{b}\\
    &x_{2} =  \bm{a}^T\kappa_3(\kappa_2([-\bm{p}_{i},-\bm{q}_{j}]))+{b}\\
    &s_{ij} = \mathbb{I}(\sigma(x_{1} + x_{1}\cdot x_{2}) + \beta z - 0.5)\\
\end{split}
\end{equation}
where $\kappa_3(x) = x + 0.5$ if $x<0$, otherwise $\kappa_3(x) = 0$. 
$\mathbb{I}(x)$ is an indicator function, which is 1 if $x>0$, and 0 otherwise.
In the experiment, we follow~\cite{zou2019focused,DBLP:conf/icml/BicaAS20} to build the testing set by uniformly recommending different items to each user.

For real-world experiments, we evaluate our model based on the following datasets:

\textbf{Amazon}\footnote{http://jmcauley.ucsd.edu/data/amazon/} is an e-commerce dataset including user feedback on products from multiple categories.

\textbf{Anime}\footnote{https://www.kaggle.com/CooperUnion/anime-recommendations-database} is crawled from myanimelist.net, which includes the user preference on a large number of animes.

\textbf{Yahoo! R3}\footnote{https://webscope.sandbox.yahoo.com/catalog.php?datatype=r} is an online music recommendation dataset, which contains the user survey data and ratings for randomly selected songs between August 22, 2006 and September 7, 2006.

\textbf{Coat}\footnote{https://www.cs.cornell.edu/~schnabts/mnar/} is a commonly used dataset for evaluating debiased recommender models.

\textbf{PCIC}\footnote{https://competition.huaweicloud.com/information/1000041488/introduction} is a recently released dataset for debiased recommendation, where we are provided with the user preferences on uniformly exposed movies.

Evaluating debiased recommender models should be based on uniform testing sets, which has been provided in {Yahoo! R3}, {Coat} and {PCIC}.
For {Amazon} and {Anime}, we follow the previous work~\cite{saito2020asymmetric,zheng2021disentangling} to resample the original testing sets to simulate uniform data according to the inverse item frequency.
The statistics of all the above datasets are summarized in Table~\ref{rec-dataset}, where we can see they can cover different application domains, and their densities vary a lot. By experimenting on these datasets, we hope to fairly evaluate different models, and demonstrate the generality of our idea.

\begin{table}[t]
\centering
\small
\caption{{Summarization of the datasets.}}
\vspace{-0.3cm}
\scalebox{1}{
\begin{threeparttable} 
\begin{tabular}{p{2.1cm}<{\centering}|p{1.7cm}<{\centering}|p{1.7cm}<{\centering}|p{1.7cm}<{\centering}|p{2.6cm}<{\centering}|p{1.5cm}<{\centering}}
\hline\hline
       Dataset               &\# User   &\# Item  &Density &Domain&Type \\ \hline
Synthetic &10000 &32   &1.756\%&-&Full\\
Yahoo! R3  &14,877&1,000&0.812\%&Music&No\\
Coat       &290   & 300 &5.333\%&Clothing&No\\
PCIC       &1000  & 1720&0.241\%&Movie&No\\
Amazon     &1,323,884     &61,275    &0.003\%&E-commerce&Semi\\
Anime      &73,516     &12,294    &0.86\%&Anime&Semi\\
\hline
\end{tabular}
\begin{tablenotes}    
   \footnotesize              
   \item[1] ``Full'', ``Semi'' and ``No'' indicate fully synthetic, semi-synthetic (only simulating the test set) and real-world datasets.
\end{tablenotes} 
\end{threeparttable}  
}
\label{rec-dataset}
\end{table}

\subsubsection{Baselines.} 
The following representative baselines are selected in our experiments:

\textbf{Direct Method (Direct)~\cite{direct1}} is a sample generation model, where the sample space are imputed before model optimization. In this model, the label of counterfactual data is estimated by a model learned from observational data. 

\textbf{Inverse Propensity Score (IPS)~\cite{DBLP:journals/corr/SwaminathanJ15}} is a distribution adjusting model, where the samples are re-weighted to convert the observed data distribution to the ideal unbiased one. The method is achieved by calculating propensity score over observational data firstly. 

\textbf{Self-nomorlized IPS (SNIPS)~\cite{DBLP:conf/nips/SwaminathanJ15}} is an extension of IPS, where the learning objective variance is reduced by a normalization strategy. 

\textbf{Doubly Robust (DR)~\cite{DBLP:conf/icml/DudikLL11}} is a combination between IPS and Direct Method. 

\textbf{ATT}~\cite{DBLP:conf/sigir/Saito20} is an extension of the direct method based on meta-learning, which leverages two imputation models to double check the correctness of counterfactual labels. 

\textbf{CVIB}~\cite{DBLP:conf/nips/WangCWHKZ20} is a debiased method based on information bottleneck.

In the experiments, we denote our model by \textbf{CBR}, which is short for \underline{\textbf{c}}onfounder \underline{\textbf{b}}alancing based \underline{\textbf{r}}ecommender method.

\begin{table*}[!t]
\caption{\small{Overall comparison between our framework and the baselines.}}
\vspace{-0.3cm}
\small
\center
\renewcommand\arraystretch{1.2}
\setlength{\tabcolsep}{6.pt}
\begin{threeparttable}  
\scalebox{0.7}{
\begin{tabular}{p{1.8cm}<{\centering}|
p{1.5cm}<{\centering}|
p{1.3cm}<{\centering}p{1.3cm}<{\centering}
p{1.3cm}<{\centering}p{1.3cm}<{\centering}|
p{1.3cm}<{\centering}p{1.3cm}<{\centering}
p{1.3cm}<{\centering}p{1.3cm}<{\centering}}\hline\hline

\multicolumn{2}{c|}{Dataset}
&\multicolumn{4}{c|}{GMF}
&\multicolumn{4}{c}{MLP}\\
\hline
\multicolumn{2}{c|}{Metrics}       
&\multicolumn{1}{c}{{NDCG@10}}&\multicolumn{1}{c}{{Recall@10}}
&\multicolumn{1}{c}{{AUC}}&\multicolumn{1}{c|}{{ACC}}
&\multicolumn{1}{c}{{NDCG@10}}&\multicolumn{1}{c}{{Recall@10}}
&\multicolumn{1}{c}{{AUC}}&\multicolumn{1}{c}{{ACC}}\\\hline

\multirow{9}{*}{Synthetic}
&Base & 0.8479 & 0.5364 & 0.8184 & 0.7408 & 0.8777 & 0.5822 & 0.7465 & 0.7552 \\ 
&IPS & 0.8578 & 0.5524 & 0.8102 & 0.7531 & 0.8991 & 0.575 & 0.7542 & 0.7496 \\ 
&SNIPS & 0.8634 & 0.5514 & 0.8193 & 0.7509 & 0.9016 & 0.5726 & 0.7672 & 0.7644 \\ 
&Direct & 0.8095 & 0.4398 & 0.7301 & 0.6764 & 0.8783 & 0.5355 & 0.6081 & 0.5812 \\ 
&DR & 0.8642 & 0.5787 & 0.8294 & 0.7524 & 0.8946 & 0.5641 & 0.7533 & 0.7548 \\ 
&ATT & 0.8686 & 0.5676 & 0.8219 & 0.7635 & 0.8804 & 0.5881 & 0.7559 & 0.7525 \\ 
&CVIB & 0.8875 & 0.5765 & 0.8269 & 0.7627 & 0.9004 & 0.5746 & 0.7678 & 0.7514 \\ 
&CBR(-g) & 0.897 & 0.589 & 0.8301 & 0.7631 & 0.8916 & 0.5939 & 0.7674 & 0.7506 \\ 
&CBR & \textbf{0.9034} & \textbf{0.5969} & \textbf{0.8336} & \textbf{0.7708} & \textbf{0.9047} & \textbf{0.5987} & \textbf{0.7763 }& \textbf{0.7702}
\\\hline

\multirow{9}{*}{Yahoo! R3}
&Base & 0.5162 & 0.4874 & 0.6793 & 0.5619 & 0.5089 & 0.4839 & 0.6762 & 0.5696 \\ 
&IPS & 0.6456 & 0.6001 & 0.6822 & 0.5671 & 0.6533 & 0.5889 & 0.6825 & 0.5677 \\ 
&SNIPS & 0.6404 & 0.6014 & 0.6899 & 0.569 & 0.6618 & 0.5946 & 0.6866 & 0.5701 \\ 
&Direct & 0.6361 & 0.568 & 0.6723 & 0.5766 & 0.6488 & 0.5298 & 0.6883 & 0.5801 \\ 
&DR & 0.657 & 0.5787 & 0.6849 & 0.5698 & 0.6592 & 0.5953 & 0.6922 & 0.5823 \\ 
&ATT & 0.6331 & 0.5624 & 0.6519 & \textbf{0.5835} & 0.644 & 0.5607 & 0.6609 & 0.5725 \\ 
&CVIB & 0.6564 & 0.5856 & 0.6869 & 0.5627 & 0.6614 & 0.5955 & 0.6935 & 0.5514 \\ 
&CBR(-g) & 0.6516 & 0.5758 & 0.6701 & 0.5734 & 0.6638 & \textbf{0.603} & 0.6874 & \textbf{0.5806} \\ 
&CBR & \textbf{0.6601} & \textbf{0.6027} & \textbf{0.6909} & 0.5788 & \textbf{0.6649} & 0.6015 & \textbf{0.6955} & 0.5746 
\\\hline

\multirow{9}{*}{Coat}
&Base & 0.61 & 0.6683 & 0.5778 & 0.5198 & 0.6277 & 0.6775 & 0.5926 & 0.5346 \\ 
&IPS & 0.6437 & 0.711 & 0.6224 & 0.5996 & 0.6312 & 0.685 & 0.5918 & 0.5493 \\ 
&SNIPS & 0.6461 & 0.7164 & 0.6236 & 0.5984 & 0.629 & 0.6976 & 0.6046 & 0.5751 \\ 
&Direct & 0.5362 & 0.6431 & 0.5188 & 0.5315 & 0.5012 & 0.6178 & 0.479 & 0.4014 \\ 
&DR & 0.6269 & 0.7008 & 0.6026 & 0.5729 & 0.6207 & 0.6876 & 0.5948 & 0.5228 \\ 
&ATT & 0.602 & 0.6811 & 0.5756 & 0.5985 & 0.6244 & 0.6818 & 0.5976 & 0.5311 \\ 
&CVIB & 0.6458 & 0.7102 & 0.6208 & 0.5729 & 0.6324 & 0.703 & 0.6057 & 0.5576 \\ 
&CBR(-g) & 0.656 & 0.7225 & 0.6362 & 0.6092 & 0.6314 & 0.7037 & 0.6075 & 0.5695 \\ 
&CBR & \textbf{0.6788} & \textbf{0.7344} & \textbf{0.6401} & \textbf{0.6223} & \textbf{0.6735} & \textbf{0.7318} & \textbf{0.6362} & \textbf{0.6092 }
\\\hline

\multirow{9}{*}{PCIC}
&Base & 0.4921 & 0.801 & 0.6561 & 0.641 & 0.4821 & 0.8279 & 0.6802 & 0.6537 \\ 
&IPS & 0.4808 & 0.7971 & 0.6783 & 0.6302 & 0.5381 & 0.8279 & 0.6788 & 0.6591 \\ 
&SNIPS & 0.4977 & 0.7948 & 0.6835 & 0.6331 & 0.5404 & \textbf{0.8664} & 0.6897 & 0.6508 \\ 
&Direct & 0.3152 & 0.7368 & 0.5526 & 0.6017 & 0.3125 & 0.7317 & 0.5623 & 0.6208 \\ 
&DR & 0.4987 & 0.8294 & 0.6909 & 0.6262 & 0.5139 & 0.8151 & 0.688 & 0.6468 \\ 
&ATT & 0.5053 & 0.8254 & 0.6904 & 0.6207 & 0.5381 & 0.8292 & 0.6938 & 0.6336 \\ 
&CVIB & 0.5079 & 0.8352 & 0.7036 & 0.6307 & 0.5443 & 0.8279 & 0.6936 & 0.6441 \\ 
&CBR(-g) & \textbf{0.5142} & \textbf{0.8485} & 0.7095 & 0.6247 & 0.5389 & 0.8485 & 0.7158 & 0.6635 \\ 
&CBR & 0.5033 & 0.8408 & \textbf{0.7174} & \textbf{0.6341} & \textbf{0.5661} & 0.8305 & \textbf{0.7191} & \textbf{0.665 }
\\\hline

\multirow{9}{*}{Amazon}
&Base & 0.8287 & 0.9659 & 0.5469 & 0.6441 & 0.8989 & 0.9136 & 0.5643 & 0.5762 \\ 
&IPS & 0.9344 & 0.9293 & 0.5437 & 0.6189 & \textbf{0.9826} & 0.9357 & 0.5721 & 0.6004 \\ 
&SNIPS & 0.8705 & 0.9312 & 0.5395 & 0.6448 & 0.9711 & 0.9333 & 0.5791 & 0.5988 \\ 
&Direct & 0.896 & 0.9379 & 0.5718 & 0.6751 & 0.9605 & 0.9387 & 0.552 & 0.5941 \\ 
&DR & 0.9097 & 0.9286 & 0.5467 & 0.6745 & 0.9684 & 0.9367 & 0.5546 & 0.6023 \\ 
&ATT & 0.9254 & 0.9363 & 0.6173 & 0.6881 & 0.9318 & \textbf{0.9453} & 0.5532 & \textbf{0.6098} \\ 
&CVIB & 0.8810 & 0.9348 & 0.5594 & 0.6783 & 0.9588 & 0.9337 & 0.5647 & 0.5803 \\ 
&CBR(-g) & 0.8307 & 0.9618 & 0.5529 & 0.6299 & 0.7854 & 0.9004 & 0.5788 & 0.5951 \\ 
&CBR & \textbf{0.9345} & \textbf{0.9777} & \textbf{0.6427} & \textbf{0.6908} & 0.9626 & 0.9408 & \textbf{0.5813} &0.6061 
\\\hline

\multirow{9}{*}{Anime}
&Base & 0.8828 & \textbf{0.9864} & 0.5411 & 0.8781 & \textbf{0.9910} & 0.9663 & 0.6072 & 0.8769 \\ 
&IPS & 0.8890 & 0.8506 & 0.5425 & 0.7919 & 0.9577 & 0.8617 & 0.5651 & 0.7837 \\ 
&SNIPS & 0.8908 & 0.8493 & 0.5314 & 0.8787 & 0.9612 & 0.8627 & 0.5774 & 0.7926 \\ 
&Direct & 0.9139 & 0.8543 & 0.5872 & 0.8658 & 0.963 & 0.8608 & 0.5716 & 0.8285 \\ 
&DR & 0.9001 & 0.8525 & 0.5918 & 0.8526 & 0.9591 & 0.8551 & 0.5669 & 0.8306 \\ 
&ATT & 0.9117 & 0.9273 & 0.5695 & 0.8284 & 0.9380 & 0.9289 & 0.6142 & 0.8320 \\ 
&CVIB & 0.8905 & 0.9228 & 0.5142 & 0.8203 & 0.9343 & 0.9291 & 0.5948 & 0.8316 \\ 
&CBR(-g) & 0.9279 & 0.9627 & 0.5619 & \textbf{0.8836} & 0.9602 & 0.996 & \textbf{0.6341} & 0.8809 \\ 
&CBR & \textbf{0.9857} & 0.9730 & \textbf{0.6833} & 0.8828 & 0.989 & \textbf{0.9937} & 0.6215 & \textbf{0.8811 }
\\\hline\hline

\end{tabular}
}   
\end{threeparttable}    
\begin{tablenotes}    
    \footnotesize              
    \item ``Base'' is the original model without applying any debiased method.
\end{tablenotes} 
\label{tab:overall} 
\end{table*}

\subsubsection{Implementation details.}
For Yahoo! R3, Coat and PCIC, we directly use the testing set provided in the original data to validate and evaluate different models.
For the simulation dataset, the ratio between the training and testing (including validation) sets is controlled as 3:1, where the training set is biased by the recommender model, while the testing set is unbiased.
For the other datasets, we use 20\% of each user's interactions for testing, while the others are left for training ({70\%}) and validation ({10\%}).
We evaluate different recommender models based on the well-known metrics including Recall, AUC~\cite{rendle2012bpr}, ACC~\cite{gunawardana2009survey} and NDCG.
Recall measures the overlapping between the recommended items and the ground truth.
The latter three metrics are ranking-sensitive, where the higher ranked correct items contribute more to the results. 
For each model, we generate 10 recommendations to be compared with the ground truth. 

The hyper-parameters are determined based on grid search.
In specific, the learning rate and batch size are tuned in the ranges of {{$[10^{-1}, 10^{-2}, 10^{-3}, 10^{-4}]$} and $[64,128,256,512,1024]$}, respectively.
The weighting parameter $\lambda$ is determined in $[0.001, 0.01, 0.05, 0.1, 0.5]$.
The regularization coefficients $\lambda_k$'s are all searched in $[0.001, 0.005, 0.01, 0.05]$.
The user/item embedding dimension is empirically set as 32. 
We apply our framework to different base models including \textbf{MLP} and \textbf{GMF}~\cite{he2017neural}, which are both generalized matrix factorization methods, covering different merging strategies between the user and item embeddings.
For the baselines, we set the parameters as the optimal values reported in the original paper or tune them in the same ranges as our model's. 
All the experiments are conducted based on a server with a 16-core CPU, 128g memory and RTX 5000 GPU.

\subsection{Overall Performance Comparison}
The overall comparison results can be seen in Table~\ref{tab:overall}, from which we can see:
in most cases, direct method performs worse than IPS or SNIPS.
We speculate that the recommendation dataset is quite noisy, the imputation model learned based on it can be unreliable.
Since the recommender model is learned based on a large amount of user feedback estimated from the imputation model, the performance can not be guaranteed. 
DR can usually achieve better performance than IPS, SNIPS and Direct, which agrees with the previous work~\cite{DBLP:conf/icml/DudikLL11}. Two state-of-art methods ATT and CVIB perform better than other baselines in most cases. It is because they both consider more general situation in real-world system to avoid the noise interference and have stronger theoretical guarantee. In most cases CVIB is better than ATT, it is because ATT is an extension of direct method, which depend to some extent on how well the performance of imputation model. 
Encouragingly, our model can achieve the best performance on all the datasets across different base models and evaluation metrics.
In specific, for MLP, our model can on average improve the performance of the base model by about {10.88\%, 6.87\%, 4.21\% and 4.03\%} on NDCG@10, Recall@10, AUC and ACC, respectively.
For GMF, the performance gains on these metrics are about {12.06\%, 8.27\%, 11.24\% and 5.58\%}.
This observation demonstrates the effectiveness of our idea comparing with existing mainstream debiased recommender models.
The reason can be that the baseline models usually need to introduce additional models for estimating the inverse propensity score or user feedback.
The prediction error of these models may worsen the downstream recommendation performance.
However, in our framework, all the parameters are learned in an end-to-end manner, where there is no accumulation error, and the item balancing and recommendation tasks are jointly optimized for beneficial knowledge transfer.
As a result, we can observe improved performance of our framework. We also demonstrate the reduced version of our method CBR(-g), which is same method with CBR only without confounder inference. We find that in most cases, our method CBR achieve better performance than CBR(-g). But in PCIC dataset, we find that CBR(-g) is better than CBR on NDCG and Recall under GMF base model. It is because the number of parameters in GMF is not as enough as the MLP have.

\subsection{Ablation Studies}
For better understanding our model, in this section, we conduct many ablation studies.
We compare our model with its three variants:
in \underline{\textit{CBR(clipping)}}, we use the clipping strategy proposed in section~\ref{ECB} to handle the item balancing terms.
In \underline{\textit{CBR(sampling)}}, the item balancing terms are optimized based on the sampling strategy.
In \underline{\textit{CBR(-g)}}, we do not capture the latent confounders, that is, objective~(\ref{emp-adv}) is used for model optimization.
The model parameters are set as their optimal values tuned above.
In the experiments, we report the results based on AUC, MLP and the dataset of PCIC, while similar conclusions can be drawn for the other metric, base models and datasets.
The results are presented in Figure~\ref{ab}(a). 

We find that CBR(sampling) can achieve better performance than CBR(clipping).
The reason can be that in CBR(clipping), the balancing terms (i.e., $\text{IPM}_G(\cdot; \cdot)$) are fixed in each training epoch.
The non-selected terms have no chance to be optimized, which may still reveal important signals.
In CBR(sampling), the balancing terms are chosen in a random manner, every term has an opportunity to be optimized, which can cover more important information.
From the efficiency perspective, since CBR(sampling) needs to sample item pairs in each epoch, it costs more time as indicated by the dotted line in Figure~\ref{ab}(a). 
Careful reader may also be curious about how the number of selected balancing terms (i.e., $K_1$ for CBR(clipping) and $K_2$ for CBR(sampling)) influence the recommendation performance.
To answer this question, we further conduct experiments by tuning $K_1$ and $K_2$ in the range of $\{5, 10, 15, 20, 25, 30\}$.
The results are presented in Figure~\ref{ab}(b).
We can see, the performance is in general better when we use larger $K_1$ or $K_2$, which is as expected, since more balancing requirements have been satisfied.
Actually, we have also tried to introduce all the balancing terms like objective~(\ref{emp}), but it costs too much time (more than 10 hours per epoch) to obtain the results.
By considering all the balancing terms in a feasible adversarial manner, CBR can achieve better performance and higher efficiency than both of CBR(clipping) and CBR(sampling).

The modeling of the latent confounders is also important, which is evidenced by the lowered performance of CBR(-g) comparing with CBR.
We speculate that in real-world recommendation datasets, latent confounders can be inevitable, since user personalities are too complex to be totally recorded.
Without capturing the latent confounders, the causal relation between the user and item is hard to cut down, which may intensify the distribution shift problem, and lowers the recommendation performance.

\begin{figure}[t]
\centering
\setlength{\fboxrule}{0.pt}
\setlength{\fboxsep}{0.pt}
\fbox{
\includegraphics[width=0.8\linewidth]{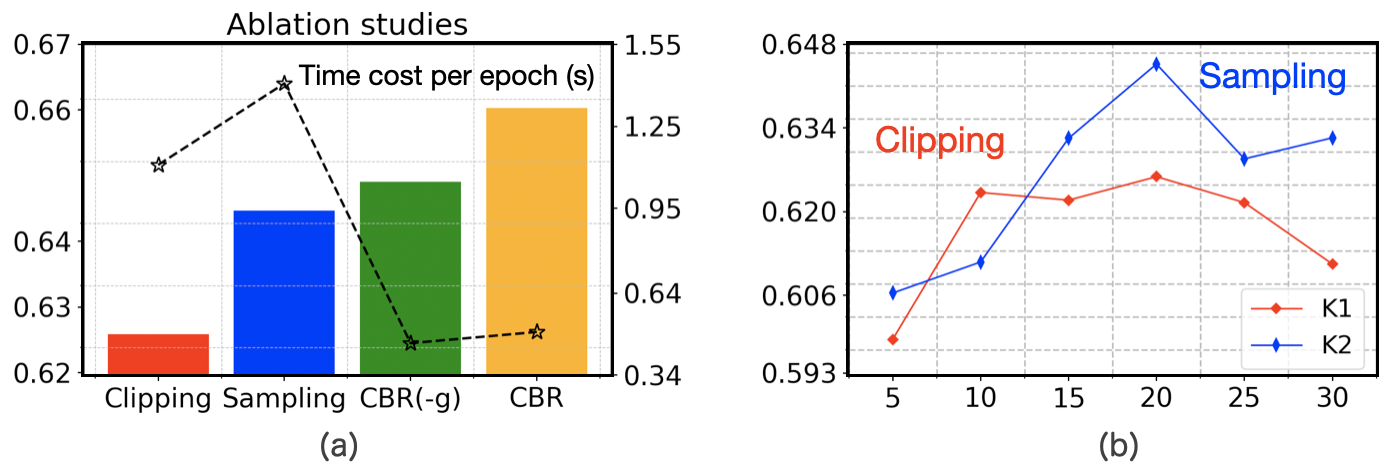}
}
\caption{
(a) Performance and time cost of different model variants.
(b) The influences of the selected balancing terms.
}
\label{ab}
\end{figure}

\begin{figure}[t]
\centering
\setlength{\fboxrule}{0.pt}
\setlength{\fboxsep}{0.pt}
\fbox{
\includegraphics[width=.5\linewidth]{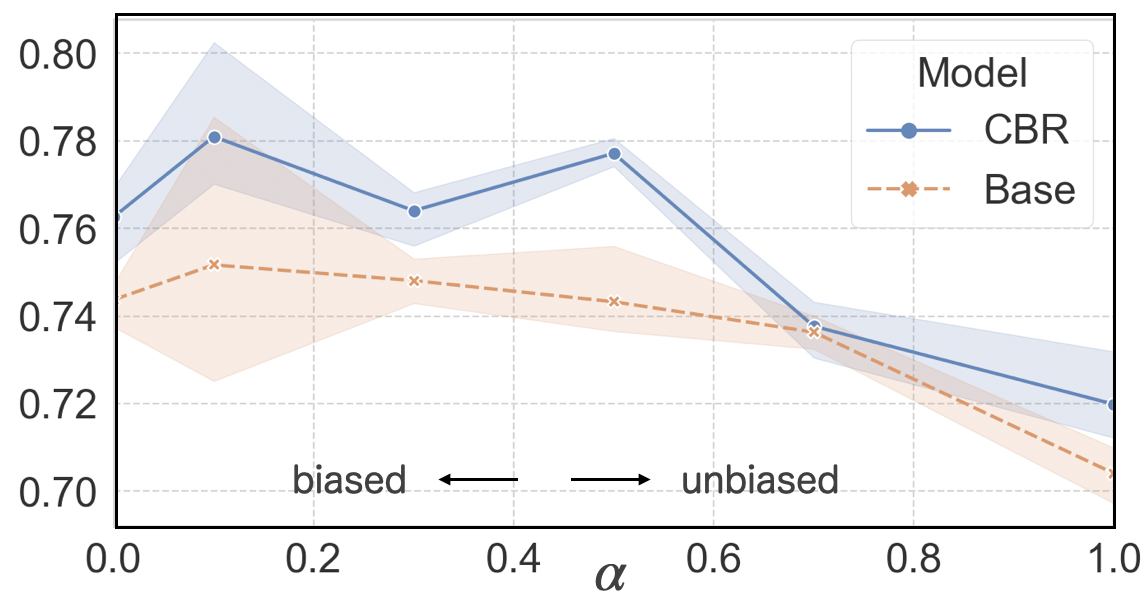}
}
\caption{
Performance of our model with different data biases.
For each $\alpha$, we repeat the experiment for {3 times}, and report the mean and standard error of the results.
}
\label{bias}
\end{figure}

\subsection{Influence of the Data Bias Severity}
In this section, we study how the data bias influence the final recommendation performance.
We base this experiment on the simulation dataset, since it allows us to flexibly tune the bias severity by the parameter $\alpha$.
Extremely, $\alpha=1.0$ means the data is unbiased, and an item is recommended in a completely random manner.
If $\alpha=0.0$, then the data is completely biased by the recommender model.
We tune $\alpha$ in the range of $\{0.0,0.1,0.3,0.5,0.7,1.0\}$, and set $\beta$ as {0.5}.
The model parameters in this experiment are as their optimal values tuned above.
Due to the space limitation, we report the results on AUC and MLP, while the conclusions on ACC and GMF are similar and omitted. 

The results are presented in Figure~\ref{bias}, from which we can see:
our model can consistently achieve better performance comparing with the base model.
This observation manifests that the effectiveness of our model holds for different data biases, and suggests that our model can be robust when being applied to different recommendation scenarios.
It is interesting to see that the performance gain of our framework is more significant when the data bias is relative large.
The reason can be that, for the base model, the parameters are directly learned based on the biased training sets.
They can not perform well on unbiased testing sets.
However, by our framework, the learning objective is specially designed to remove the bias effect, which alleviates the distribution shift between the training and testing sets, and thus can achieve better performance. 
In general, the severer the bias is, the larger the performance gap is.
Careful readers may find that when $\alpha=1.0$, there is a small gap between our framework and the base model.
The reason can be that, in the synthetic dataset, we have incorporated the latent confounders into the simulation process.
In our framework, we explicitly deploy a model to capture such confounders, while in the base model, they are totally ignored, which leads to the lowered performance.

\begin{figure}[t]
\centering
\setlength{\fboxrule}{0.pt}
\setlength{\fboxsep}{0.pt}
\fbox{
\includegraphics[width=.5\linewidth]{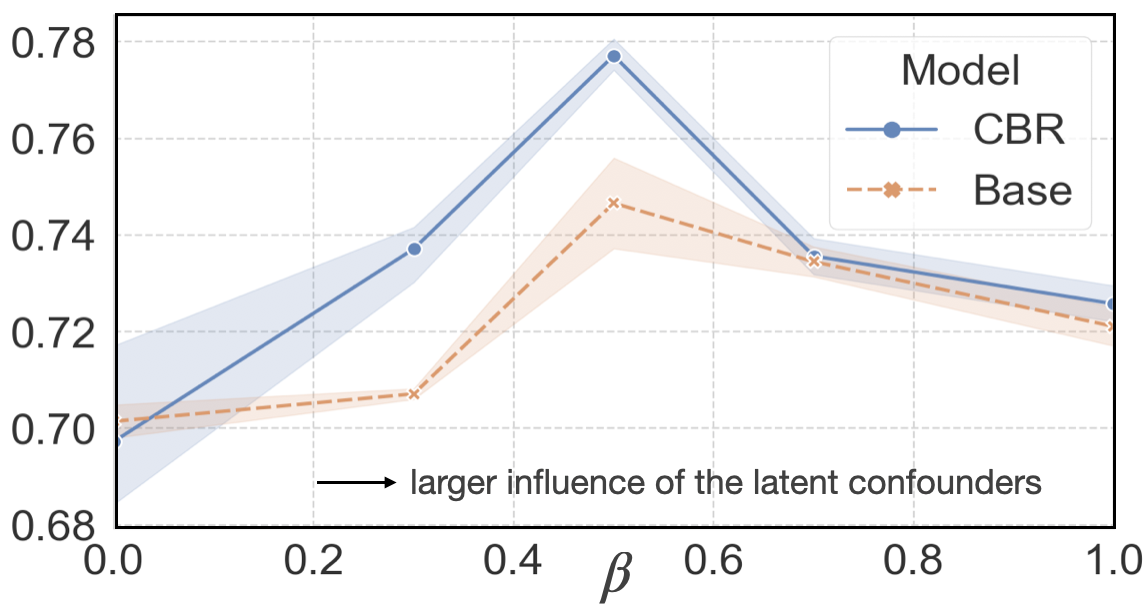}
}
\vspace{-0.cm}
\caption{
Performance of our model with different confounderness.
For each $\beta$, we repeat the experiment for {3 times}, and report the mean and standard error of the results.
}
\vspace{-0.cm}
\label{confounder}
\end{figure}

\subsection{Influence of the Confounderness}
In real-world recommender systems, there are usually many latent confounders, which are either not recorded in the system (e.g., item promotion) or hard to represent clearly (e.g., user emotions).
Such confounders can directly influence the debiasing effect.
In this section, we study the influence of latent confounders on the recommendation performance. 
The experiment is based on the simulation dataset, where $\beta$ is the parameter to control the confounderness of the dataset.
Larger $\beta$ means the dataset is more severely influenced by the latent confounders.
We tune $\beta$ in the range of $\{0.0,0.3,0.5,0.7,1.0\}$, and set $\alpha$ as 0.5.
For the other settings, we follow the above experiment.
The results are presented in Figure~\ref{confounder}, from which we can see:
when the influence of the latent confounders is very large (e.g., $\beta\geq0.7$), the performances of our framework and the base model are similar and not satisfied.
We speculate that when $\beta$ is too large, the important collaborative filtering (CF) \cite{he2018outer} signals are overwhelmed by the latent confounders.
Collaborative filtering stands on the key assumption of the recommendation task.
Without enough CF signals, neither of the compared models can achieve good performance. 
When $\beta$ falls into a reasonable range (e.g., $\beta\in [0.3,0.5]$), we can observe significant improvement of our framework comparing with the base model.
This observation demonstrates the effectiveness of our deconfounder idea, and suggests that our framework can be competitive for practical recommendation scenarios, which usually contain latent confounders.

\begin{figure}[t]
\centering
\setlength{\fboxrule}{0.pt}
\setlength{\fboxsep}{0.pt}
\fbox{
\includegraphics[width=.5\linewidth]{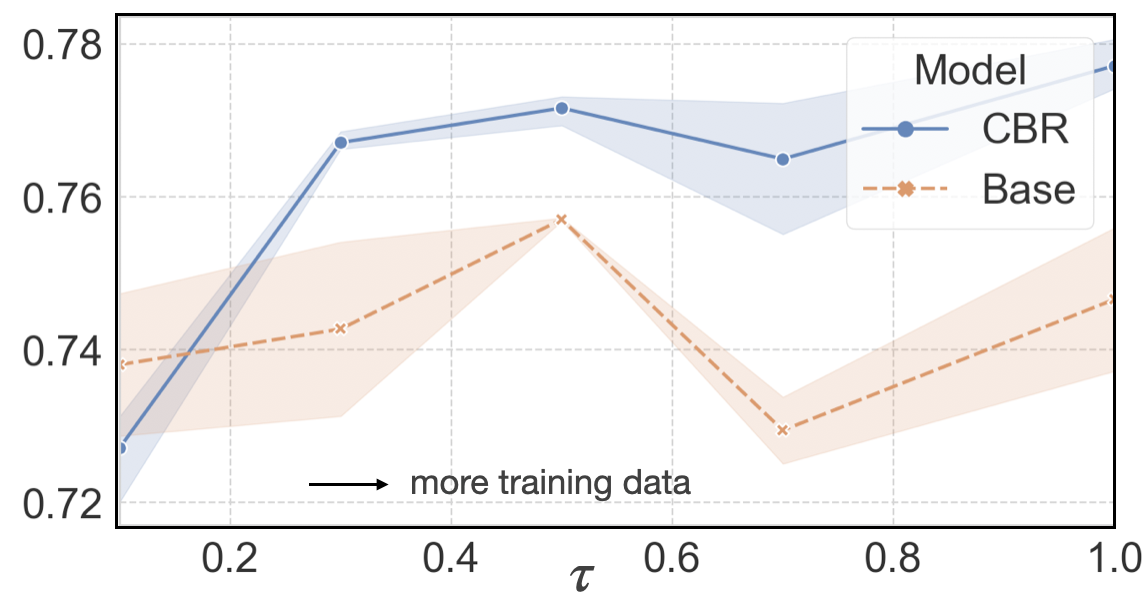}
}
\vspace{-0.cm}
\caption{
Performance of our model with different data sparsities.
For each $\tau$, we repeat the experiment for {3 times}, and report the mean and standard error of the results.
}
\vspace{-0.cm}
\label{sparsity}
\end{figure}

\subsection{Influence of the Data Sparsity}
In this section, we study the influence of data sparsity on the recommendation performance.
The settings of this experiment is similar to the above experiment, and the synthetic dataset is generated by setting $\alpha$ and $\beta$ as 0.5.
We use parts of the original training set for model optimization, and the testing set remains to be unbiased.
The ratio $\tau$ between the selected and original number of samples is tuned in $\{0.1, 0.3, 0.5, 0.7, 1.0\}$, where $\tau=1.0$ means incorporating all the training samples. 
The results are presented in Figure~\ref{sparsity}, from which we can see:
as we use more training samples, the performance of our framework continually goes up, which agrees with the learning theory that more samples can lower the bound between the empirical and expectation loss functions.
When the number of samples is relative small (e.g., $\tau\leq0.5$), all the models exhibit similar worse performances.
We speculate that, in this case, data sparsity is a dominate problem. 
As the old saying goes, ``even a clever housewife cannot cook a meal without rice''.
Without a basic amount of training samples, any model cannot perform well.
When the dataset becomes denser, the data bias comes to the major issue.
At this time, our framework can achieve much larger performance gains as compared with the base model.

\section{Conclusion}
In this paper, we propose a novel debiased recommender model based on confounder balancing.
We begin from the ideal unbiased learning objective, and derive its upper bound under the effect of a user representative function.
In order to enhance the training efficiency, we design three methods to handle the large number of item balancing terms.
We also model the potential latent confounders, which can lead to more robust optimization.

This paper actually opens a new door for debiased recommendation.
There is much room left for improvement.
To begin with, one can assume more complex confounder structures to involve reasonable prior knowledge to improve the recommendation performance.
In addition, people can also extend our idea to sequential or even graph-based recommendation, where the user-item structure information is considered into the modeling process.

\bibliographystyle{ACM-Reference-Format}
\bibliography{acmart}
\clearpage

\end{document}